%% file: main.tex
\documentclass[twocolumn,journal,10pt]{IEEEtran}
\usepackage{amsmath,amssymb,amsfonts}
\usepackage{amsthm}
\usepackage{mathtools}
\usepackage[table]{xcolor} 
\usepackage{booktabs} 
\usepackage{eqparbox}
\usepackage{float}
\usepackage{algorithm}
\usepackage{longtable}
\usepackage[noend]{algpseudocode}
\usepackage{subfig}
\usepackage{xcolor}
\usepackage[utf8]{inputenc}
\usepackage{caption}
\usepackage{textcomp}
\usepackage{lettrine}
\usepackage{algpseudocode}
\usepackage{xspace}
\usepackage{multirow}
\usepackage{acronym}
\usepackage{array}   
\usepackage{booktabs} 
\usepackage{caption} 
\usepackage{subcaption} 
\usepackage[short]{optidef}
\usepackage{subfig}
\input{acronyms}
\definecolor{lighterviolet}{RGB}{220,200,240}
\usepackage[backend=biber,style=ieee]{biblatex} 
\usepackage{hyperref} 
\hypersetup{
    colorlinks=true, 
    linkcolor=blue,  
    citecolor=blue,  
    urlcolor=cyan    
}
\makeatletter
\renewcommand\@cite[2]{\textcolor{blue}{[}\textcolor{blue}{#1}\textcolor{blue}{]}} 

 \newtheorem{lemma}{Lemma}

\begin{document}

\title{A Novel Framework for Near-Field Covert Communications with RIS and RSMA}

\author{Atiquzzaman Mondal, Amira Bendaimi, Hüseyin Arslan~\IEEEmembership{Fellow,~IEEE}

\thanks{A. Mondal, A. Bendaimi and H. Arslan are associated with the Department of Electrical and Electronics Engineering at Istanbul Medipol University, Kavacık North Campus, 34810 Istanbul, Türkiye (email: atiq336@gmail.com, amira.bendaimi@std.medipol.edu.tr, huseyinarslan@medipol.edu.tr).}

}

{}

\maketitle

\begin{abstract}
This paper explores the near field (NF) covert communication with the aid of rate-splitting multiple access (RSMA) and reconfigurable intelligent surfaces (RIS). In particular, the RIS operates in the NF of both the legitimate user and the passive adversary, enhancing the legitimate user's received signal while suppressing the adversary’s detection capability. Whereas, the base station (BS) applies RSMA to increase the covert communication rate composed of a private and a shared rate component. To characterize system covertness, we derive closed-form expressions for the detection error probability (DEP), outage probability (OP), and optimal detection threshold for the adversary. We formulate a non-convex joint beamforming optimization problem at the BS and RIS under unit-modulus constraints to maximize the covert rate. To tackle this, we propose an alternating optimization (AO) algorithm, where the BS beamformer is designed using a two-stage iterative method based on successive convex approximation (SCA). Additionally,  two low-complexity techniques are introduced to further reduce the adversary’s received power. Simulation results demonstrate that the proposed algorithm effectively improves the covert communication rate, highlighting the potential of near field RSMA-RIS integration in covert communication. 

\end{abstract}

\begin{IEEEkeywords}
Reconfigurable intelligent surfaces (RIS), covert communication, rate-splitting multiple access (RSMA) , near-field (NF), hybrid beamforming.
\end{IEEEkeywords}

\section{Introduction}
\lettrine{W}{ith} the development of beyond 5G (B5G) wireless technology, the demand for high speed and ultra reliable data transmission has risen exponentially in the modern society, which have been facilitated by the adoption of XL-MIMO systems, millimeter wave (mmWave), terahertz (Thz) technology, and reconfigurable intelligent surfaces RISs \cite{jiang2024physical,alsabah20216g}. However, security remains a
critical concern in these deployments, due to the broadcast nature of the shared medium, which introduces various privacy and security vulnerabilities such as eavesdropping and jamming attacks. For instance, Not always that a receiver can get what is being transmitted from the transmitter side, an unintended user often termed as Warden/Willie tries to decode those without being noticed \cite{hamamreh2018classifications}.

In recent years, Covert communication has emerged as cutting-edge security technology that can be used to conceal the confidential message undetected \cite{chen2023covert}. various techniques such as multi-antenna system, full duplex, relay have been used to achieve positive covert communication \cite{chen2021multi, shu2019delay, lv2022achieving}. However, since these techniques are usually channel-adaptive, the wireless propagation environment has a significant impact on covert communication. Henceforth, RIS is considered a promising solution for providing more controlled and secure communication \cite{zhang2025star}
\begin{table*}[t!]
\centering
\caption{Comparison table with some existing literature}
\begin{tabular}{c|cccccc}
\hline
\textbf{Reference} & \textbf{RIS} & \textbf{Near-field (NF)} & \textbf{RSMA} & \textbf{KPIs}\\ \hline
{\cite{wu2022passive}} & $\checkmark$ & $\times$ & $\times$ & OP, intercept probability (IP), DEP and system security probability (SSP) \\ \hline
{\cite{zhu2023active}} & $\checkmark$ & $\times$ & $\times$ & covert rate, OP, DEP\\ \hline
{\cite{liang2025covert}} & $\checkmark$ & $\times$ & $\checkmark$ & covert communication rate, OP, DEP\\ \hline
{\cite{zhang2025star}} & $\checkmark$ & $\times$ & $\checkmark$ & covert communication rate\\ \hline
{\cite{10609798}} & $\checkmark$ & $\checkmark$ & $\times$ & covert rate, beamforming design\\ \hline
Our proposed work & $\checkmark$ & $\checkmark$ & $\checkmark$ & OP, DEP, covert communication rate, beamforming design \\ \hline
\multicolumn{5}{l}{$\checkmark$: considered; $\times$: not considered.}
\end{tabular}
\label{Comparison table}
\end{table*}

RIS is an evolving technology consisting of passive reflecting elements, controllable by software and is capable of manipulating the electromagnetic (EM) wave by dynamically reconfiguring the propagation environment by adjusting the amplitude and phase for the reflected signal \cite{wu2019beamforming, trichopoulos2022design, huang2022reconfigurable}, typically used in blockage issues to improve the network coverage saving significant power \cite{barbuto2021metasurfaces, pan2021reconfigurable, basharat2021reconfigurable, liang2021reconfigurable}. 
To reap the benefits of RIS the authors in \cite{kong2021intelligent, wu2022passive,zhu2023active,lu2020intelligent} incorporated RIS into covert communication systems by optimizing the reflection matrix and jointly designing transmit beamformers. However, beamforming design plays a critical role in covert communication, as it not only strengthens the signal for legitimate users but also reduces signal leakage to potential eavesdroppers such as Willie. Therefore, this design becomes more challenging when aiming to balance improved link quality for the legitimate user and reduced detectability.

The current existing literatures on RIS-assisted covert communications presents several shortcomings. In particular, at THz, mmWave or supposedly beyond those frequency ranges that has significant path loss, the double fading effect issue becomes more severe \cite{yu2020wideband,xing2021millimeter,rodriguez2018frequency}. Furthermore, most prior works focus primarily on far field (FF) communications under planar wave propagation, which leads to limited covert rate performance in spatially correlated environments. As the signal carrying beam broadens up, unintended users might be able to detect the information and it becomes very difficult to differentiate between the received signals of the passive adversary and the legitimate user as they are positioned closely with respect to the BS.

Lately, near field RIS has emerged as a key enabler for enhancing covert communication in the NF region that becomes operational as both the array aperture and operating frequencies increase. The EM propagation in the NF is described by the spherical wave (SW) channel model, which depends on both angle and distance \cite{liu2024near, want2011near}, unlike the traditional FF planar wave channel model. By enabling beam focusing, it allows the beam pattern to focus on a particular area, thereby the transmission via the NF-RIS creates an effective covert communications \cite{ramzan2023reconfigurable, wu2023enabling, zheng2024location}. Nevertheless, concurrently serving multiple users while maintaining covert transmission and achieving high spectral efficiency remains a significant challenge.  

In this aspect, another emerging technology i.e., RSMA is deemed not only more effective but also a robust multiple access technique for successful multi-user communication \cite{10993434, bash2012square, bash2014lpd}. In \cite{cai2021resource}, the authors proposed a cooperative rate-splitting (CRS) transmission scheme to enhance physical layer security, where the common messages are used as artificial noise to confound potential eavesdroppers. With a view to boost the energy and spectrum efficiency, the authors in \cite{mao2018energy} examined RSMA with other multiple access techniques. In \cite{zhou2021rate}, the authors suggested that RSMA outperforms traditional techniques in terms of spectrum and energy efficiency in a multi-antenna transmission scenario. Additionally, the works in \cite{zhang2024rate} demonstrated the effectiveness of RSMA in covert communication, especially in enhancing the rate of the covert communication.

Moreover, the integration of RSMA with other cutting edge technologies such RIS has gained a substantial attention among researchers from industry and academia to make the wireless communication system ultra reliable. Specifically, by regulating the phase shifts and amplitude of the reflected signals, RIS introduces a new approach for covert communications and makes it challenging for the adversaries to detect the communication activities with an improved covert performance \cite{kong2021intelligent, wu2023irs}. In \cite{liang2025covert,zhang2025star,shambharkar2022rate}, the authors look into a RIS-assisted RSMA communication system, where the BS communicates with multiple users via RSMA in order to enrich covertness. Nevertheless, these studies primarily focus on far-field RIS-assisted RSMA systems, the potential of near-field RSMA remains largely unexplored and worth investigate.

In summary, the effectiveness of covert communication in B5G communication system is expected to be greatly influenced by the emerging technologies such as NF communication, RIS, and RSMA. We have presented a comparison in Table~\ref{Comparison table} whereby just one or two of these technologies have been used in the framework of covert communication, therefore underlining the originality of this work. To the best of the authors’ knowledge, the application of RSMA to near field RIS systems for covertness analysis has not yet been studied in the literature. In this work, analyze the covertness performance and optimization of near field RSMA aided by RIS. The main contributions of this study are outlined as follows

\begin{itemize}
        \item We introduce a novel near-field covert communication framework using RIS and RSMA, where the RIS is placed in the NF of both the covert user i.e., Bob and Willie to enhance Bob’s received power while attenuating Willie’s, thereby strengthening system covertness.
        
        \item We formulate the DEP expressions for Willie, the covert communication rate, as well as the optimal threshold for it's detection in closed form to build a covertness constraint for optimization. Furthermore, we analyzed the outage probability for the system under consideration. 
         
       \item We formulate a non-convex optimization problem for the joint design of BS beamforming and RIS phase shifts aimed at maximizing the covert rate, accounting for unit-modulus constraints and interdependent large-scale variables associated with the RIS, public users, and Willie.  
       
        \item We develop an alternating-optimization (AO) algorithm that splits the problem into two subproblems. The BS beamformer is designed via a two-phase iterative method leveraging successive convex approximation (SCA). To further minimize the complexity of the method used, we propose two suboptimal techniques so that Willie's received signal power is minimized.
         
       \item Simulation results confirm the robustness and superior security performance of the proposed NF-RSMA scheme over its far-field counterpart, demonstrating higher average covert rates and lower probabilities of detection. 
   
\end{itemize}
The reminder of this work is arranged as follows. Section II outlines the system model. The DEP performance at Willie is analyzed in Section III along with  the closed-form equations for DEP, the transmission rate, and OP of the covert communication system. Section IV illustrates the convergence analysis along with the optimization methods. The analytical results are then explained in Section V, which is followed by conclusion of the work in Section VI.

The notations used in this paper are given in Table~\ref{Notations}.

\begin{table}[h!]
\centering
\caption{Notations used in this work}
\resizebox{\columnwidth}{!}{
\begin{tabular}{cc}
\hline
\textbf{Notations} & \textbf{Definitions} \\
\midrule
\rowcolor{blue!10} $N = N_{y} \times N_{z}$ & Number of RIS reflective elements with $N_{y}$ and $N_{z}$ being the horizontal rows and vertical columns respectively \\
$\mathcal{CN}(\mu, \sigma^{2})$ & The complex AWGN distribution with mean $\mu$ and variance $\sigma^{2}$ \\
\rowcolor{blue!10} $\mathbb{E}(\cdot)$ and $\text{Var}(\cdot)$ & Expectation and variance \\
$(\cdot)^{T}, (\cdot)^{*}, (\cdot)^{H}$ & Transpose, conjugate, and the Hermitian of a quantity \\
\rowcolor{blue!10} $||\cdot||_{2}$ & Denotes the Euclidean norm \\
$vec(\cdot), Tr(\cdot)$ & Denote the vector operation and trace of any quantity \\
\rowcolor{blue!10} $|\cdot|$ & Denotes the absolute value \\
$\circ, \otimes$ & Denote the Hadamard product and Kronecker product \\
\rowcolor{blue!10} $\text{Pr}(\cdot)$ & Denotes the probability \\
$\text{diag}(\cdot)$ & Denotes the diagonal matrix of a quantity \\
\rowcolor{blue!10} $\mathbb{C}^{X \times Y}$ & Denotes the set of $X \times Y$ complex matrix \\
$f_{\sigma_{w}^{2}}(\cdot)$ & Denotes the PDF of Willie \\
\rowcolor{blue!10} $\mathbf{H}_{br}, \mathbf{g}_{rb}, \mathbf{g}_{rw}, \mathbf{h}_{bru}$ & Denote the channel links BS-RIS, RIS-Bob, RIS-Willie, and BS-RIS-User respectively \\
$\mathbf{y}_{k}, \mathbf{y}_{W}$ & Denote the received signal at User-k and at Willie respectively \\
\rowcolor{blue!10} $P_{bs}, P_{W}$ & Denotes the transmit power of BS and average received power of Willie respectively \\
$R_{c,k}, R_{p,k}$ & Denote the achievable rates for decoding $s_{c}$ and $s_{k}$ respectively \\
\rowcolor{blue!10} $\gamma_{c,k}, \gamma_{p,k}$ & Denotes the SINR for the common and private signal respectively \\
$\mathbf{W}, \Theta$ & Denote the BS beamforming matrix and RIS phase shift matrix respectively \\
\bottomrule
\end{tabular}
}
\label{Notations}
\end{table}
\vspace{-8mm}
\section{System Model}
As illustrated in Fig.~\ref{sys_model}, we consider a covert communication system enabled by a NF RSMA assisted by RIS, consisting of a legitimate BS, $K$ RSMA users denoted as $\left(\text{User}~1, \text{User}~2, \cdots, \text{User}~K\right)$, where $\text{User}~1$ (referred as Bob) is the covert user, while the remaining users are public users and an illegal user named Willie acts as the warden, attempting to detect the presence of any data transmissions from the BS. In this particular configuration, Bob, other users and Willie are equipped with a single antenna and BS is equipped with $M_{BS}$ antennas $\left(M_{BS} \geq 1\right)$\footnote{The BS operates in high-frequency bands (e.g., mmWave or THz), requiring hybrid beamforming; a fully digital architecture is assumed at Bob for analytical simplicity.}. We assume that the direct link between BS and Bob is blocked, and the indirect link is aided by an RIS that is strategically placed within the NF region and consists of $N = N_y \times N_z$ passive reflecting elements arranged in a uniform planar array (UPA), where $N_{y}$ and $N_{z}$ denote numbers of elements in horizontal and vertical directions, respectively. Additionally, The phase shift of the RIS is given as $\mathbf{\Theta}=diag\left(e^{j\theta_{1}},\cdots,e^{j\theta_{N}}\right)$, where $\theta_{n}\in [0,2\pi)$ is the phase shift of $n^{th}$ reflecting element with its amplitude assumed to be unity. 
\vspace{-6mm}
\subsection{Channel Model}
In this setup, the widely adopted Saleh-Valenzuela (SV) channel model is used to characterize the channels between the BS and various nodes, including the RIS, Bob, Willie, and ~$k$-th public user. Specifically, the BS–RIS channel is denoted by $\mathbf{H}_{br} \in \mathbb{C}^{N \times M_{BS}} \sim \mathcal{CN}\left(0, \sigma_{br}^{2} \mathbf{I}_{N} \right)$, while the RIS–user $x$ channels are represented by $\mathbf{g}_{rx} \in \mathbb{C}^{1 \times N} \sim \mathcal{CN}\left(0, \sigma_{rx}^{2} \mathbf{I}_{N} \right)$, where the subscript  $ x \in \{b,w,u_k\}$ refers to Bob, Willie, and the $k$-th public user, respectively. Since the RIS is assumed to be located within FF region of the BS, the BS–RIS channel $\mathbf{H}_{br}$ is modeled using a far-field planar-wave-based channel representation, defined as \cite{bendaimi2024leverage}
\begin{align}\label{eqn_4}
    \mathbf{H}_{br} = \tilde{P}_{bt} \sum_{j=1}^{J_p} \daleth_j \mathbf{a}_{\text{UPA}} (\vartheta_{\text{AoA}}^{j}, \phi^{j}) \mathbf{a}_{\text{ULA}}^{\text{H}} (\vartheta_{\text{AoD}}^{j}),
\end{align}
where $\tilde{P}_{bt}=\sqrt{\frac{M_{BS} N D_{pl}}{J_p}}$, $(J_p)$ is the total resolvable signal paths, $(D_{pl})$ is the average path loss, $(\daleth_j)$ and $(\vartheta_{\text{AoD}}^{j})$ denote the $(j)^{th}$ path's complex channel gain and the angle of departure (AoD) associated with BS respectively.  $(\phi^{j})$ and $(\vartheta_{\text{AoA}}^{j})$ denote the azimuth angle and elevation angles of arrival (AOAs) related to the RIS. Meanwhile, $\mathbf{a}_{\text{UPA}}$ and $\mathbf{a}_{\text{ULA}}$ represent the normalized UPA vector and the uniform linear array (ULA) vector, respectively.
Thus, $\mathbf{a}_{\text{ULA}}$ is defined as
\begin{align}\label{eqn_5}
    &\mathbf{a}_{\text{ULA}} (\vartheta)=\nonumber\\
    & \frac{1}{\sqrt{M}} \left[ 1, \ldots, e^{j \frac{2\pi d}{\lambda} (M-1) \sin(\vartheta)}, \ldots, e^{j \frac{2\pi d}{\lambda} (M-1) \sin(\vartheta)} \right]^{\text{T}},
\end{align}
where $d$ is the separation of the antenna and $M$ is the number of elements of the ULA. Similarly, $\mathbf{a}_{\text{UPA}}$ is given as \cite{bendaimibeam}
\begin{align}\label{eqn_6}
    &\mathbf{a}_{\text{UPA}} (\vartheta, \phi) = \frac{1}{\sqrt{N}}\mathbf{a}_{y}\left(\vartheta, \phi \right)\otimes\mathbf{a}_{z}\left( \phi \right)\nonumber\\
    &=\frac{1}{\sqrt{N}} \left[ 1, e^{j \frac{2\pi d}{\lambda} \sin(\vartheta) \cos(\phi)}, \ldots, e^{j \frac{2\pi d}{\lambda} (N_{y} - 1) \sin(\vartheta) \cos(\phi)} \right]^{\text{T}}\nonumber\\
    &\otimes \left[ 1, e^{j \frac{2\pi d}{\lambda} \sin(\phi)}, \ldots, e^{j \frac{2\pi d}{\lambda} (N_{z}-1) \sin(\phi)} \right]^{\text{T}}.
\end{align}
\begin{figure}[t!]
    \centering
        \includegraphics[width=\linewidth]{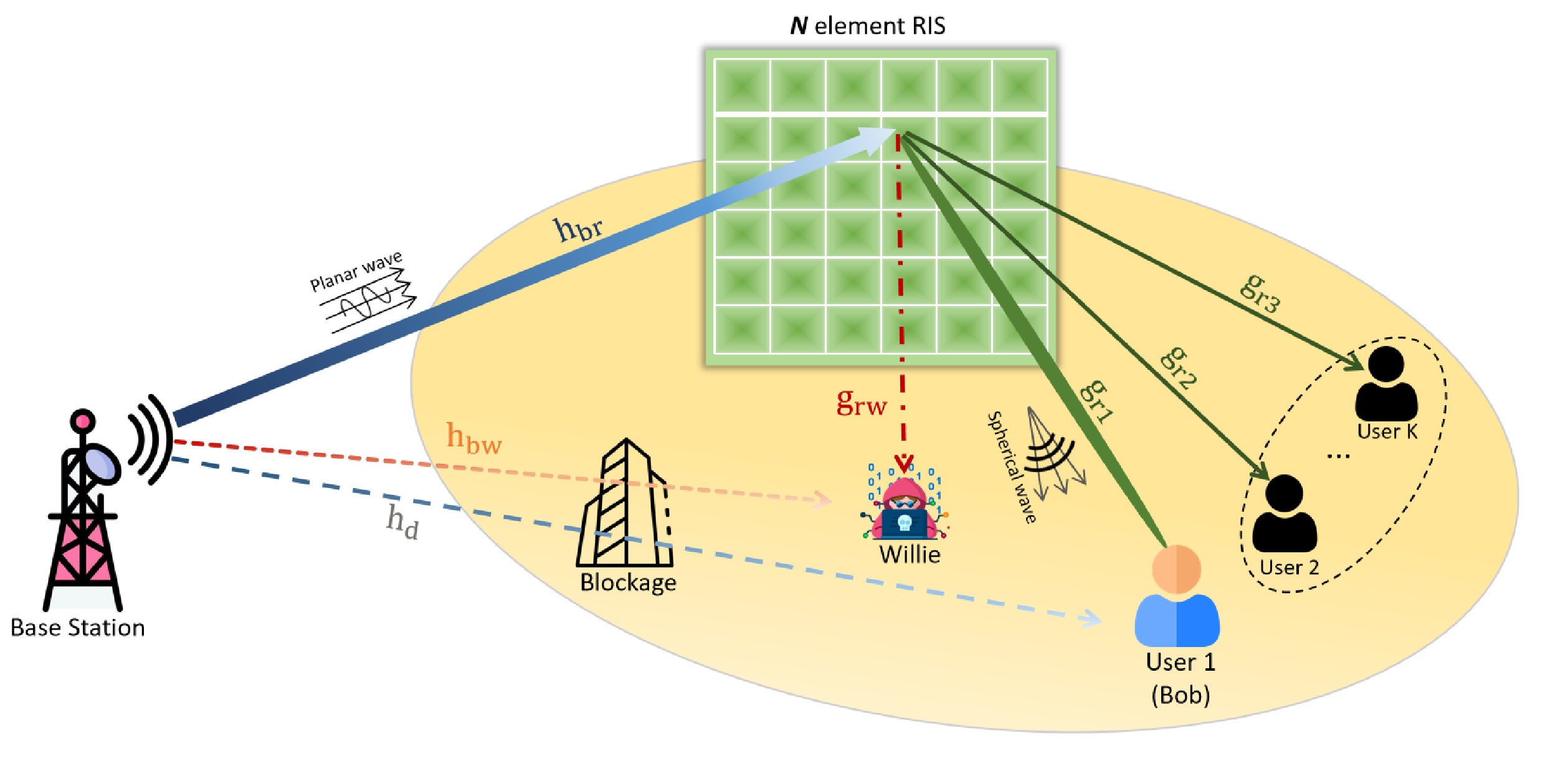}
    \caption{Illustration of a NF-RIS-RSMA-aided covert communication system.}
   \label{sys_model}
\end{figure}

Since the separation between the RIS and the user $x$ is shorter than the Rayleigh distance\footnote{We assume all users, including Willie to be in the NF due to significant received power discrepancies, especially for Willie. Suppose, Willie is in the FF and Bob in NF, then the received power of Willie could be lower than the noise power due to double-fading in higher frequencies, making the considered scenario inefficient. Conversely, if Bob is FF and Willie NF, it eventually becomes RIS-assisted covert communication in FF, which is not this work's focus. The NF exists if $0.62\sqrt{\frac{D_{RIS}}{\left( c/f\right)}}<d_{Rayleigh}<\frac{2D_{RIS}^{2}}{\left( c/f\right)}$, which implies spherical signal propagation after reflection from the RIS since $d_{Rayleigh}$ is proportional to $D_{RIS}$ and frequency band used.} $\left(d_{Rayleigh}= 2\left(D_{RIS} + D_{U_x} \right)*f/c \right)$, where $c$, $f$, $D_{RIS}$ and $D_{U_x}$ are speed of light, carrier frequency, RIS antenna aperture and antenna aperture of user $x$, respectively. Bob, Willie and other public users are assumed to be located within the near-field region of the RIS. Henceforth, the channel  $\mathbf{g}_{rx} \mathbb{C}^{N\times 1}$ is modeled using a near-field spherical-wave representation, given by \cite{lu2023near}
\begin{equation}
   \mathbf{g}_{rx}=\sqrt{\frac{N}{L_x}} \sum_{l=1}^{L_x} \alpha_l \mathbf{b}\left(\theta_{l_x}, r_{l_x}\right), 
\end{equation} where $L_{x}$ denote the number of paths between RIS and user $x$, $\alpha_{l}$ is the complex gain,  and $\mathbf{b}\left(\theta_{l_x}, r_{l_x}\right)$ $\in \mathbb{C}^{N\times 1}$, represent the near-field array response vector. without loss of generality,  $\mathbf{b}\left(\theta_l, r_l\right)$ can be defined as follows \cite{lu2023near}
\begin{equation}
    \mathbf{b}\left(\theta_l, r_l\right)=\frac{1}{\sqrt{N_1}}\left[e^{-j \frac{2 \pi}{\lambda}\left(r_l^{(1)}-r_l\right)}, \cdots, e^{-j \frac{2 \pi}{\lambda}\left(r_l^{\left(N\right)}-r_l\right)}\right]^H
\end{equation} where $r_l$ denotes the distance of the $l$-th scatterer from the center of the transmitter antenna array, $r_l^{\left(n\right)}=$ $\sqrt{r_l^2+\delta_{n}^2 d^2-2 r_l \delta_{n} d \sin \theta_l}$ represents the distance of the $l$-th scatterer from the $n$-th transmitter antenna, $\delta_{n}=$ $\frac{2 n_-N_-1}{2}$ with $n=1,2, \cdots, N, \theta_l \in(-\pi / 2, \pi / 2)$ are the practical physical angles and $d$ represents half a wavelength of antenna spacing. Furthermore, the overall cascaded BS-to-user $x$ channel via RIS can be explicitly given as
\begin{align}\label{eqn_8}
   \mathbf{h}_{brx} = \mathbf{H}_{br}^{H}\mathbf{\Theta}\mathbf{g}_{rx}.
\end{align}
\vspace{-9mm}
\subsection{Signal Model}
 To address the message that has been transmitted from the BS to all the users, we incorporate the RSMA framework, via which the BS will cleave the information into common and private segment. However, the common message corresponding to the $k$ users are obfuscated into a common message stream, while the private messages are merged into a private message stream. Therefore, the signal that the BS transmits is written as
\begin{align}\label{eqn_1}  
    \mathbf{x}_{bs} = \underbrace{\mathbf{w}_{c} s_{c}}_{\text{Common signal}} + \underbrace{\sum_{k=1}^{K} \mathbf{w}_{k} s_{k}}_{\text{Private signal}},
\end{align}
where $\mathbf{w}_c \in \mathbb{C}^{M_{BS} \times 1}$ and $\mathbf{w}_{k} \in \mathbb{C}^{M_{BS} \times 1}$ are the linear beamforming vectors for the common and private symbols, respectively, and $\mathbb{E}\{|s_{c}|^2\} = 1$ and $\mathbb{E}\{|s_{k}|^2\} = 1$.

We consider that the warden (Willie), adopts an energy detection technique to identify possible transmission from the BS. For covert communication, the covert rate is greatly impacted by noise calibration mistakes introduced at Willie by variations in the electromagnetic environment. Adopting the constraint uncertainty model, the probability density function (PDF) of Willie's noise power $\sigma_w^2$ can be represented as 
\begin{align}\label{eqn_2}
    f_{\sigma_w^2}(x_{bs}) = 
        \begin{cases} 
            \frac{1}{2 \ln (\rho) x_{bs}}, & \frac{1}{\rho} \tilde{\sigma}_w^2 \leq x_{bs} \leq \rho \tilde{\sigma}_w^2 \\ 
            0, & \text{otherwise},
        \end{cases}
\end{align}
where $\tilde{\sigma}_w^2$ indicates the power of noise not having any ambiguity, i.e., $\rho = 1$, and $\rho$ represents uncertainty coefficient of the noise. The following inequality is satisfied by Willie's noise power, which is given as 
\begin{align}\label{eqn_3}
    \frac{1}{\rho} \tilde{\sigma}_w^2 \leq \sigma_w^2 \leq \rho \tilde{\sigma}_w^2.
\end{align}

Therefore, the signal that is received at User $k$, i.e., $\mathbf{y}_{k}$\footnote{It is important to note that, when $k=1$, it refers to the RIS-Bob link i.e., $\mathbf{g}_{rb}$ or $\mathbf{y}_{b}$.} can be represented as
\begin{align}\label{eqn_9}
   \mathbf{y}_{k} = \mathbf{h}_{bru}^{H}\mathbf{w}_{c} s_{c}+ \sum\limits_{k=1}^{K}\mathbf{h}_{bru,k}^{H}\mathbf{w}_{k} s_{k}+n_{k},
\end{align}
where $n_{k} \sim \mathcal{CN}\left(0, \sigma_{k}^{2}\right)$. In accordance with the RSMA protocol, the user incipiently decodes the common message $s_{c}$ while treating the private message as interference. Following this, the user then interprets its own private message stream, treating other users' private message streams as interference, after performing successive interference cancellation (SIC) to remove $s_{c}$. Therefore, the approximate achievable rates for decoding $s_{c}$ and $s_{k}$ for the considered user $k$ can be respectively conveyed as
\begin{align}\label{eqn_10}
    R_{c,k} &= \text{log}_{2}\left(1+\frac{|\mathbf{h}_{bru}^{H} \mathbf{w}_c |^2}{ \sum\limits_{k=1}^{K} |\mathbf{h}_{bru}^{H} \mathbf{w}_{k} |^2 + \sigma_{k}^{2} }\right), \quad \forall k \in K,
\end{align}
and
\begin{align}\label{eqn_11}
     R_{p,k} &= \text{log}_{2}\left(1\frac{|\mathbf{h}_{bru}^{H} \mathbf{w}_{k} |^2}{ \sum\limits_{i=1, i\neq k}^{K} |\mathbf{h}_{bru}^{H} \mathbf{w}_{i} |^2 + \sigma_{k}^2 }\right), \quad \forall k \in K.
\end{align}
One of the important factors in assessing the performance of system under consideration is the covert rate. According to the RSMA principle, the common rate cannot exceed the channel capacity associated with the common symbol so as to guarantee that every user considered in the system can properly decode the common message. The segment of the $k^{th}$ user related to the common rate is $\sum\limits_{k=1}^{K} \wp_{c,k} \leq R_{c,k}, \forall k \in K$. Therefore, the overall possible sum-rate of the considered RSMA under the NF system can therefore be given as 
\begin{align*}
    R_{\text{Total}} = \sum\limits_{k=1}^{K} \left(\wp_{c,k} + R_{p,k}\right).
\end{align*}
\section{Covertness Performance analysis}
This section analyses the covertness requirements of the proposed system. Usually Willie looks for evidence that the BS is sending covert information to Bob or not. For this, Willie utilizes the Neyman-Pearson criterion or binary hypothesis testing problem to make the decisions. $(\mathcal{K}_{0})$ i.e., the null hypothesis indicates that the BS is silent and did not send any message to Bob, whereas $(\mathcal{K}_{1})$ i.e., the alternative hypothesis suggests that the BS is covertly communicating with Bob. Accordingly, the received signal at Willie under the assumption of $\mathcal{K}_{0}$ and $\mathcal{K}_{1}$ can be given as follows
\begin{align}\label{eqn_12}
    \mathbf{y}_{W}= 
        \begin{cases} 
           \mathbf{n}_{w}; \quad \mathcal{K}_{0} \\ 
             \left(\mathbf{H}_{br}^{H}\mathbf{\Theta}\mathbf{g}_{rw}\right)\mathbf{x}_{bs}+\mathbf{n}_{w}; \quad \mathcal{K}_{1}
        \end{cases}
\end{align}
where $n_{w}$ is the additive  Gaussian noise (AWGN), i.e., $n_{w} \sim \mathcal{CN}\left(0, \sigma_{w}^{2}\right)$. Accordingly, Willie's average received power can be expressed as
\begin{align}\label{eqn_13}
    P_{W}= 
        \begin{cases} 
             \sigma_{w}^{2}; \quad \mathcal{K}_{0} \\ 
            P_{bs}\left|\left|\mathbf{H}_{br}^{H}\mathbf{\Theta}\mathbf{g}_{rw} \right|\right|_{2}^{2} + \sigma_{w}^{2}; \quad \mathcal{K}_{1}.
        \end{cases}
\end{align}
Subsequently, the covert performance of the proposed system is analyzed.First, the DEP at Willie is derived, followed by the analysis of the user's outage probability (OP) and covert transmission rate.
\subsection{Detection Error Probability (DEP) at Willie}
This subsection presents the analytical expressions for the probability that Willie successfully detects the covert message transmitted by the BS to Bob, i.e., the hypothesis test performed by Willie, which is quantified by the detection error probability (DEP), defined as
\begin{align}\label{eqn_14}
   P_{DEP}=P_{FAP}+P_{MDP}, \quad 0 \leq P_{DEP} \leq 1,
\end{align}
where $P_{FAP} = Pr\left(\mathcal{D}_{1} | \mathcal{K}_{0}\right) = Pr\left(P_{W} > \zeta| \mathcal{K}_{0}\right)$ is the probability that Willie incorrectly detects a transmission when BS does not actually sent any messages, and $P_{MDP} = Pr\left(\mathcal{D}_{0} | \mathcal{K}_{1}\right) = Pr\left(P_{W} < \zeta| \mathcal{K}_{1}\right)$ is the missed detection probability in which Willie incorrectly detects that BS has not sent any messages when it has already sent. Here, $\zeta$, $\mathcal{D}_{0}$ and $\mathcal{D}_{1}$ denote the threshold of the power detection of Willie, decisions of Willie that the BS is not transmitting or is transmitting to Bob, respectively. Specifically, $P_{DEP}=0$ indicates Willie's perfect detection of error free covert transmission from BS to Bob, while $P_{DEP}=1$ indicates the Willie's failure in detecting the covert transmission and gives a random assumption. 

Furthermore, based on the criterion of Neyman-Pearson hypothesis, the most effective strategy for Willie to achieve the minimum DEP is the likelihood test i.e.,
$\frac{p_{1}\left(\mathbf{y}_{W}\right)}{p_{0}\left(\mathbf{y}_{W}\right)}  \overset{\scriptstyle \mathcal{D}_{1}}{\underset{\scriptstyle \mathcal{D}_{0}}{\gtrless}}1$, where $p_{0}\left(\mathbf{y}_{W}\right)$ and $p_{1}\left(\mathbf{y}_{W}\right)$ are the corresponding likelihood functions related to the received signals for the hypotheses $\mathcal{K}_{0}$ and $\mathcal{K}_{1}$, respectively. Therefore, depending on the \eqref{eqn_13} and \eqref{eqn_14}, the expression for $P_{FAP}$ and $P_{MDP}$ for Willie can be respectively expressed as
\begin{align}\label{eqn_15}
  P_{FAP} = \Pr\left(P_{W} > \zeta | \mathcal{K}_0\right) = \Pr\left(\sigma_{w}^{2} > \zeta | \mathcal{K}_0\right),
\end{align}
and 
\begin{align}\label{eqn_16}
   P_{MDP} = \Pr\left(P_{W} < \zeta | \mathcal{K}_1\right) = \Pr\left(\aleph + \sigma_{w}^{2} < \zeta | \mathcal{K}_1 \right),
\end{align}
where $\aleph = P_{bs}\left|\left|\mathbf{H}_{br}^{H}\mathbf{\Theta}\mathbf{g}_{rw} \right|\right|_{F}^{2}$. Therefore, the total DEP can be expressed as
\begin{align}\label{eqn_17}
   P_{DEP}&=\Pr\left(P_{W} > \zeta | \mathcal{K}_0\right) + \Pr\left(P_{W} < \zeta | \mathcal{K}_1\right)\nonumber\\
   &= 1- \Pr \left(\zeta-\aleph \leq \sigma_w^2 \leq \zeta \right)\nonumber\\
   &= 1 - \int\limits_{\max\left( \zeta - \aleph, \frac{1}{\rho} \tilde{\sigma}_w^2 \right)}^{\min\left( \zeta, \rho \tilde{\sigma}_w^2 \right)} f_{\sigma_w^2}(x_{bs}) \, dx_{bs}\nonumber\\
   &= 1-\frac{1}{2\text{ln}(\rho)}\left\{\text{ln}\left[\min\left( \zeta, \rho \tilde{\sigma}_w^2 \right)\right] \right.\nonumber\\
   &-\left.\text{ln}\left[\text{max}\left( \zeta - \aleph, \frac{\tilde{\sigma}_{w}^{2}}{\rho}\right) \right]\right\}
\end{align}
\begin{lemma}\label{Lem1}
     Consider $\zeta \geq \rho \tilde{\sigma}_w^2 $. The the optimal detection threshold i.e., $\zeta^{*}$ is constrained by $\rho \tilde{\sigma}_w^2$.  
\end{lemma}
\begin{proof}
    From \eqref{eqn_2} and \eqref{eqn_13}, based on the hypothesis $\mathcal{K}_{0}$ the maximum averaged received power at Willie, can be given as $\rho \tilde{\sigma}_w^2$. 

   Furthermore, under the hypotheses $\mathcal{K}_{0}$ and $\mathcal{K}_{1}$, the probability distributions Willie's average received power i.e., $P_{W}$ would overlap. If the distributions do not overlap then the DEP i.e., $\zeta = 0$. Hence, it is important to control the leakage power $P_{bs}\left|\left|\mathbf{H}_{br}^{H}\mathbf{\Theta}\mathbf{g}_{rw} \right|\right|_{F}^{2}$ within a definite range. Subsequently, the optimal detection threshold $\zeta^{*}$ is constrained by $\zeta \leq \rho \tilde{\sigma}_w^2$.
\end{proof}
Next, taking the derivatives of \eqref{eqn_17} with respect to $\zeta$, we obtain
\begin{align}\label{eqn_18}
    \frac{\partial}{\partial\zeta}P_{DEP}\left(\zeta \right)&=  \frac{\partial}{\partial\zeta}\left\{ 1-\frac{1}{2\text{ln}(\rho)}\left\{\text{ln}\left[\min\left( \zeta, \rho \tilde{\sigma}_w^2 \right)\right]  \right.\right.\nonumber\\
   &-\left.\left.\text{ln}\left[\text{max}\left( \zeta - \aleph, \frac{\tilde{\sigma}_{w}^{2}}{\rho}\right) \right]\right\}\right\}\nonumber\\ &=
        \begin{cases} 
             -\frac{1}{2\text{ln}(\rho)\zeta}; \quad if \quad \zeta < \frac{\tilde{\sigma}_w^2}{\rho}+\aleph \\ 
             -\frac{1}{2\text{ln}(\rho)}\left(\frac{1}{\zeta}- \frac{1}{\zeta-\aleph}\right); \quad if \quad \zeta \geq \frac{\tilde{\sigma}_w^2}{\rho}+\aleph.
        \end{cases}
\end{align}
From the conditions in \eqref{eqn_18}, the optimal detection threshold $\zeta^{*}$ is determined as
\begin{align}\label{eqn_19}
    \zeta^{*}=\text{min}\left(\frac{\tilde{\sigma}_w^2}{\rho}+\aleph, \rho \tilde{\sigma}_w^{2} \right).
\end{align}
To find the minimum DEP denoted as $P_{DEP}^{min}$ at Willie, we substitute \eqref{eqn_19} into \eqref{eqn_17}, yielding
\begin{align}\label{eqn_20}
P_{DEP}^{min}=
\begin{cases} 
             1 - \frac{1}{2\text{ln}(\rho)}\text{ln}\left(1+\frac{\rho\aleph}{\tilde{\sigma}_w^2} \right);\quad if \quad \aleph < \tilde{\sigma}_w^2\left(\rho - \frac{1}{\rho}\right)\\
             0;\quad if \quad \aleph \geq \tilde{\sigma}_w^2\left(\rho - \frac{1}{\rho}\right).
        \end{cases}
\end{align}        
Therefore, we can conclude that, when $P_{DEP}^{min}=0$, i.e., at $\aleph \geq \tilde{\sigma}_w^2\left(\rho - \frac{1}{\rho}\right)$, Willie can clearly differentiate between $\mathcal{K}_{0}$ and $\mathcal{K}_{1}$. Therefore, we are more concerned about the scenario where Willie literally finds it difficult to distinguish the hypotheses, i.e., at
\begin{align}\label{eqn_21}
   \aleph < \tilde{\sigma}_w^2\left(\rho - \frac{1}{\rho}\right).
\end{align}        
For any value $\varsigma \in [0,1]$ defining the level of covertness, the $P_{DEP}$ must satisfy $P_{DEP}>1-\varsigma$ to fulfill the requirements for covertness. As a result, the received power at Willie should satisfy
\begin{align}\label{eqn_22}
   \varsigma < \frac{(e^{2\varsigma \ln(\rho)} - 1) \tilde{\sigma}_w^2}{\rho}.
\end{align} 
The maximum permitted leakage power at Willie under this covertness constraint can be obtained from \eqref{eqn_21} and \eqref{eqn_22} as
\begin{align}\label{eqn_23}
    P_{leak} = \min \left\{\tilde{\sigma}_w^2 \left( \rho - \frac{1}{\rho} \right), \frac{(e^{2\varsigma \ln(\rho)} - 1) \tilde{\sigma}_w^2}{\rho} \right\}.
\end{align} 
It is noteworthy that the parameter $ \varsigma$ usually falls within a narrow range, namely $(\varsigma < 0.5)$. From \eqref{eqn_2}, we knew out that $\rho \geq 1$. In accordance with \eqref{eqn_23}, we can observe that the first term consistently surpasses the second term, which suggests that the leakage power is actively constrained by the covertness requirement. Therefore, within a specific range, $P_{leak}$ shows a monotonically increasing behavior with respect to $\varsigma$ under the given conditions, which implies that the amount of leakage power to Willie decreases with the increase in the covertness requirement. Therefore, the covertness constraint for the system design can be given as:
\begin{align}\label{eqn_24}
    P_{leak} > P_{bs}\left|\left|\mathbf{H}_{br}^{H}\mathbf{\Theta}\mathbf{g}_{rw} \right|\right|_{F}^{2}.
\end{align} 
\subsection{Outage Probability}
According to the RSMA principle, each user first decodes the transmitted common message $s_c$ while treating all private messages as interference. Subsequently, the user decodes its intended private message $s_k$ (for $i \neq k$), assuming perfect SIC of the common message. Based on \eqref{eqn_10} and \eqref{eqn_11}, the SINRs for the common and private signals at user $k$ can be respectively expressed as
\begin{align}\label{eqn_25}
    \gamma_{c,k} = \frac{|\mathbf{h}_{bru,k}^H \mathbf{w}_c|^2}{\sum\limits_{i=1}^K |\mathbf{h}_{bru,k}^H \mathbf{w}_i|^2 + \sigma_k^2},
\end{align}
and
\begin{align}\label{eqn_26}
    \gamma_{p,k} = \frac{|\mathbf{h}_{bru,k}^H \mathbf{w}_k|^2}{\sum\limits_{i \neq k}^K |\mathbf{h}_{bru,k}^H \mathbf{w}_i|^2 + \sigma_k^2}.
\end{align}

For Bob ($k=1$)
\begin{align}\label{eqn_27}
    \gamma_{c,1} = \frac{|\mathbf{h}_{bru,1}^H \mathbf{w}_c|^2}{\sum\limits_{i=1}^K |\mathbf{h}_{bru,1}^H \mathbf{w}_i|^2 + \sigma_1^2},
\end{align}
\begin{align}\label{eqn_28}
    \gamma_{p,1} = \frac{|\mathbf{h}_{bru,1}^H \mathbf{w}_1|^2}{\sum\limits_{i=2}^K |\mathbf{h}_{bru,1}^H \mathbf{w}_i|^2 + \sigma_1^2}.
\end{align}
In general, an outage occurs for user $k$ if the outage occurs in either the common or the private signal, i.e., if $\gamma_{c,k} < \gamma_{th,c}$ or $\gamma_{p,k} < \gamma_{th,p}$, where $\gamma_{th,c}$ and $\gamma_{th,p}$ are the threshold SINR related to the common and private message segment, respectively. Therefore, the outage probability for user $k$ is defined as
\begin{align}\label{eqn_29}
    P_{\text{out},k} &= \Pr\left(\left(\gamma_{c,k} < \gamma_{thc}\right) \cup \left(\gamma_{p,k} < \gamma_{thp}\right)\right)\nonumber\\
    &= 1 - \Pr\left(\gamma_{c,k} \geq \gamma_{th,c}, \gamma_{p,k} \geq \gamma_{th,p}\right)
\end{align}
Since the events $\gamma_{c,k} < \gamma_{th,c}$ and $\gamma_{p,k} < \gamma_{th,p}$ are correlated i.e., both depend on \eqref{eqn_8}. Therefore, the approximate OP of the User $k$ is given as
\begin{align}\label{eqn_30}
    P_{\text{out},k} & \approx F_{\gamma_{c,k}}(\gamma_{th,c}) + F_{\gamma_{p,k}}(\gamma_{th,p}) \nonumber\\
    &- \min\left\{F_{\gamma_{c,k}}(\gamma_{th,c}), F_{\gamma_{p,k}}(\gamma_{th,p})\right\},
\end{align}
where $F_{\gamma_{c,k}}(\gamma) = \Pr(\gamma_{c,k} \leq \gamma)$ and $F_{\gamma_{p,k}}(\gamma) = \Pr(\gamma_{p,k} \leq \gamma)$ are the cumulative distribution functions (CDFs) of the SINRs of common and private mmessages respectively. Based on the \eqref{A.15} and \eqref{A.16} (refer Appendix A), the OP can be given as in \eqref{eqn_31}, which is given on top of the next page.
\begin{figure*}[t!]
{
\begin{align}\label{eqn_31}
    P_{\text{out},1} &\approx \left( 1 - e^{-\frac{\gamma_{thc} \sigma_1^2}{\alpha_c P_{bs} \sigma_{bru,1}^2}} \prod_{i=1}^K \frac{1}{1 + \frac{\gamma_{thc} \alpha_i}{\alpha_c}} \right)+ \left( 1 - e^{-\frac{\gamma_{thp} \sigma_1^2}{\alpha_1 P_{bs} \sigma_{bru,1}^2}} \prod_{i=2}^K \frac{1}{1 + \frac{\gamma_{thp} \alpha_i}{\alpha_1}} \right) - \min \left\{ F_{\gamma_{c,1}}(\gamma_{thc}), F_{\gamma_{p,1}}(\gamma_{thp}) \right\},
\end{align}

}
where
\begin{align*}
    \sigma_{bru,1}^2 = \sigma_{br}^2 |\beta|^2 \frac{\lambda_1}{N}, \quad \lambda_1 = \sum_{n=1}^N \frac{1}{(4\pi \frac{r_{1n}}{\lambda})^2}.
\end{align*}
\normalsize
\hrulefill
\end{figure*}
\section{Joint Beamforming Optimization}
This subsection presents the analysis for maximizing the achievable covert transmission rate of our proposed system by jointly designing the BS beamforming matrix i.e., $\mathbf{W} = \left[\mathbf{w}_c, \mathbf{w}_1, \cdots, \mathbf{w}_K\right]$ and RIS phase-shift i.e., $\Theta$ for both the common message segment as well as the private messages segment. Depending on the mathematical formulation derived before, the optimization problem can be formulated as \(\mathcal{P}0\), which is given as
\begin{maxi!}|s|%
    {\wp_{c}, \mathbf{W}, \psi}{\sum\limits_{k=1}^{K} \left( \wp_{c,k} + R_{p,k} \right)}%
    {\label{32a}}%
    {}
    \addConstraint{\sum_{i=1}^K \wp_{c,i} \leq \log_2 \left( 1 + \gamma_{c,k} \right), \quad \forall k \in K \label{32b}}
    \addConstraint{\wp_{c,k} + R_{p,k} \geq R_k^{\text{min}}, \quad \forall k \in K \label{32c}}
    \addConstraint{R_{p,k} \leq \log_2 \left( 1 + \gamma_{p,k} \right), \quad \forall k \in K \label{32d}}
    \addConstraint{\operatorname{Tr} \left( \mathbf{W}^H \mathbf{W} \right) \leq P_{\text{max}} \label{32e}}
    \addConstraint{\left| \Theta_n \right| = 1, \quad n \in N \label{32f}}
    \addConstraint{\wp_{c,k} \geq 0, \quad \forall k \in K \label{32g}}
\end{maxi!}
where $\wp_c = \{ \wp_{c,k} \}$ and the constraints from \eqref{32b}-\eqref{32g} correspond to the common rate constraint, the QoS constraint which ensures the minimum required rate, approximate achievable rate, BS transmit power constraint, the unit modulus constraint and covert rate constraint respectively. The formulated optimization problem is quite arduous to solve due to the highly coupling variables as well as the intractable unit modulus constraints. To tackle this issue, an AO-based solution was adopted to address the coupling of variables and the non-convexity of the objective function.
\vspace{-10mm}
\subsection{Transmit Beamforming Design at the BS}
For the considered RIS phase shift matrix i.e., $\mathbf{\Theta}$, the sub-optimal value of $\mathbf{W}$ can be obtained using the AO approach under the fixed phase shift at the RIS, i.e., $\Theta = \tilde{\Theta}$. Therefore, the $\mathcal{P}0$ can be reformulated as \(\mathcal{P}1\) by incorporating the auxiliary variables $\mathbf{\Upsilon}_p$, $\boldsymbol{\beth}_p$, and $\mathbf{\nu}_c$
\begin{maxi!}|s|%
    {\mathbf{W}, \wp_c, \mathbf{\Upsilon}_p, \boldsymbol{\beth}_p, \mathbf{\nu}_c} {\sum_{k=1}^K \left( \wp_{c,k} + \Upsilon_{p,k} \right)}%
    {\label{33a}}%
    {}
    \addConstraint{\sum_{i=1}^K \wp_{c,i} \leq \log_2 \left( 1 + \nu_{c,k} \right), \quad \forall k \in K \label{33b}}
    \addConstraint{\wp_{c,k} + \Upsilon_{p,k} \geq R_k^{\text{min}}, \quad \forall k \in K \label{33c}}
    \addConstraint{\Upsilon_{p,k} \leq \log_2 \left( 1 + \beth_{p,k} \right), \quad \forall k \in K \label{33d}}
    \addConstraint{\eqref{eqn_25}, \eqref{eqn_26}, \eqref{32e}, \eqref{32g}}\notag
\end{maxi!}
where $\mathbf{\Upsilon}_p = [\Upsilon_{p,1}, \cdots, \Upsilon_{p,K}]$, $\boldsymbol{\beth}_p = [\beth_{p,1}, \cdots, \beth_{p,K}]$, $\mathbf{\nu}_c = [\nu_{c,1}, \cdots, \nu_{c,K}]$, $\overline{\mathbf{h}}_{bru}$ is substituted as $\overline{\mathbf{h}}_{bru} \triangleq \mathbf{h}_{bru} |_{\Theta = \tilde{\Theta}}$. A closed-form solution for $\mathcal{P}1$ is not feasible due to the non concavity of \eqref{33a} along with the non convexity nature of the constraints in \eqref{eqn_25} and \eqref{eqn_26}. In order to solve this, we iteratively approximate $\mathcal{P}1$ into a convex form using the SCA technique. Furthermore, auxiliary variables $\mathbf{\kappa}_c = [\kappa_{c,1}, \cdots, \kappa_{c,K}]$ and $\mathbf{\eta}_p = [\eta_{p,1}, \cdots, \eta_{p,K}]$ are introduced to simplify the terms in \eqref{eqn_25} and \eqref{eqn_26}, respectively as follows
\begin{align}\label{eqn_34}
\nu_{c,k} &\geq \frac{\left| \overline{\mathbf{h}}_{bru}^H \mathbf{w}_c \right|^2}{\kappa_{c,k}}, \quad \beth_{p,k} \geq \frac{\left| \overline{\mathbf{h}}_{bru}^H \mathbf{w}_k \right|^2}{\eta_{p,k}}, \quad \forall k \in K,  \\
\kappa_{c,k} &\leq \sum_{i=1}^K \left| \overline{\mathbf{h}}_{bru}^H \mathbf{w}_i \right|^2 + \sigma_k^2, \label{eqn_35} \\
\eta_{p,k} &\leq \sum_{i=1, i \neq k}^K \left| \overline{\mathbf{h}}_{bru}^H \mathbf{w}_i \right|^2 + \sigma_k^2. \label{eqn_36}
\end{align}
Now, using the SCA technique, the non-convex expressions in \eqref{eqn_34} can be converted into an equivalent affine form as given below
\begin{align}
\frac{2 \Re \left( \mathbf{w}_c^{(j)H} \overline{\mathbf{h}}_{bru} \overline{\mathbf{h}}_{bru}^H \mathbf{w}_c \right)}{\kappa_{c,k}^{(j)}} &- \frac{\left| \overline{\mathbf{h}}_{bru}^H \mathbf{w}_c^{(j)} \right|^2 \kappa_{c,{bru}}}{\left( \kappa_{c,k}^{(j)} \right)^2} \nonumber\\&\geq \nu_{c,k}, \quad \forall k \in K, \label{eqn_37} \\
\frac{2 \Re \left( \mathbf{w}_d^{(j)H} \overline{\mathbf{h}}_{bru} \overline{\mathbf{h}}_{bru}^H \mathbf{w}_k \right)}{\eta_{p,k}^{(j)}} &- \frac{\left| \overline{\mathbf{h}}_{bru}^H \mathbf{w}_k^{(j)} \right|^2 \eta_{p,k}}{\left( \eta_{p,k}^{(j)} \right)^2} \nonumber\\&\geq \beth_{p,k}, \quad \forall k \in K, \label{eqn_38}
\end{align}
where the values obtained at the $j^{\text{th}}$ SCA iteration are $\mathbf{w}^{(j)}, \mathbf{w}_c^{(j)}, \kappa_c^{(j)}$, and $\eta_p^{(j)}$. Consequently, the transmit optimization beamforming problem $\mathcal{P}1$ using $j^{\text{th}}$ SCA iteration can be ultimately given as $\mathcal{P}2$
\begin{maxi!}|s|%
    {\mathbf{W}, \wp_c, \mathbf{\Upsilon}_p, \boldsymbol{\beth}_p, \mathbf{\nu}_c, \mathbf{\kappa}_c, \mathbf{\eta}_p} {\sum_{k=1}^K \left( \wp_{c,k} + \Upsilon_{p,k} \right)}%
    {\label{39a}}%
    {}
    \addConstraint{\eqref{32e}, \eqref{32g}, \eqref{33b}, \eqref{33c}, \eqref{33d}}{\notag}%
    \addConstraint{\eqref{eqn_35}, \eqref{eqn_36}, \eqref{eqn_37}, \eqref{eqn_38}.}{\notag}%
\end{maxi!}
Therefore, we can say that $\mathcal{P}2$ is converted into a convex optimization problem with all the considered optimization variables. As a result, we suggest an iterative approach to the transmit beamforming design using the convergence factor $\epsilon$, as described in Algorithm~\ref{Algo_1}.
\subsection{RIS Reflection Beamforming Optimization}
For the obtained $\mathbf{W} = \overline{\mathbf{W}}$, the optimal RIS beamformer $\mathbf{\Theta}$ is designed. For this, the optimization problem can be expressed in a similar fashion by introducing the auxiliary variable $\hat{\mathbf{\Upsilon}}_p, \hat{\boldsymbol{\beth}}_{p}$, and $\hat{\mathbf{\nu}}_{c}$. Therefore, the $\mathcal{P}0$ will be reformulated as as the optimization problem $\mathcal{P}3$, given as
\begin{maxi!}|s|%
    {\Theta, \hat{\wp}_{c}, \hat{\mathbf{\Upsilon}}_{p}, \hat{\boldsymbol{\eta}}_{p}, \hat{\mathbf{\nu}}_{c}}{\sum\limits_{k=1}^{K} \left( \hat{\wp}_{c,k} + \hat{\Upsilon}_{p,k} \right)}%
    {\label{40a}}%
    {}
    \addConstraint{\sum_{k=1}^{K} \wp_{c,k} \leq \log_{2} \left( 1 + \hat{\nu}_{c,k} \right), \quad \forall k \in K \label{40b}}
    \addConstraint{\hat{\wp}_{c,k} + \hat{\Upsilon}_{p,k} \geq R_{k}^{\text{min}}, \quad \forall k \in K \label{40c}}
    \addConstraint{\hat{\Upsilon}_{p,k} \leq \log_2 \left( 1 +  \hat{\beth}_{p,k} \right),\forall k \in K \label{40d}}
    \addConstraint{ \hat{\wp}_{c,k} \geq 0, \quad \forall k \in K \label{40e}}
    \addConstraint{\eqref{eqn_25}, \eqref{eqn_26}, \eqref{32f}}\notag
\end{maxi!}
\begin{algorithm}[H]
\caption{Transmit Beamforming Design at the BS using SCA}
\label{Algo_1}
\begin{algorithmic}
\State \textbf{Input:} $\tilde{\Theta}, \epsilon$, and $\mathbf{W}^0$
\State \textbf{Initialize:} $j=0, R_{\text{sum}}^{[-1]}=0$
\While{$\left| R_{\text{sum}}^{(j)} - R_{\text{sum}}^{(j-1)} \right| \leq \epsilon$ or $j \leq J_{\text{sca}}^{\text{itr}}$}
\State Calculate $\kappa_c^{(j)}, \eta_p^{(j)}, \nu_c^{(j)}, \beth_p^{(j)}$
\State $j = j + 1$
\State Solve $\mathcal{P}2$ to obtain the beamforming matrix $\mathbf{W}^{(j)}$
\State Calculate $R_{\text{sum}}^{(j)} = \sum_{k=1}^K \left( \wp_{c,k} + \Upsilon_{p,k} \right)$
\EndWhile
\State \textbf{Output:} $\hat{R}_{\text{sum}} = R_{\text{sum}}^{(j)}$ and $\hat{\mathbf{W}} = \mathbf{W}^{(j)}$
\end{algorithmic}
\end{algorithm}
where $\hat{\mathbf{\Upsilon}}_p = [\hat{\Upsilon}_{p,1}, \cdots, \hat{\Upsilon}_{p,K}]$, $\hat{\boldsymbol{\beth}}_p = [\hat{\beth}_{p,1}, \cdots, \hat{\beth}_{p,K}]$, $\hat{\mathbf{\nu}}_c = [\hat{\nu}_{c,1}, \cdots, \hat{\nu}_{c,D}]$. $\hat{\wp}_c = [\hat{\wp}_{c,1}, \cdots, \hat{\wp}_{c,K}]$ is the rate allocation vector of the message related to the common segment. The RIS beamformer design in $\mathcal{P}3$ is non-convex in nature. Therefore, we adopt the SCA method at fixed $\mathbf{W} = \overline{\mathbf{W}}$ and reformulate $\mathcal{P}3$ as $\mathcal{P}4$ given as
\begin{maxi!}|s|%
    {\begin{array}{c}
        \Theta, \hat{\wp}_c, \hat{\mathbf{\Upsilon}}_p\\
        \hat{\beth}_p, \hat{\nu}_c, \hat{\kappa}_c, \hat{\eta}_p
     \end{array}}{\sum\limits_{k=1}^{K} \left( \hat{\wp}_{c,k} + \hat{\Upsilon}_{p,k} \right)}%
    {\label{41a}}%
    {}
    \addConstraint{\sum_{i=1}^K \left| \mathbf{w}_i^H \mathbf{h}_{bru} \right|^2 + \sigma_k^2    \leq \hat{\kappa}_{c,k} \label{41b}}
    \addConstraint{\sum_{i=1, i \neq k}^K  \left| \mathbf{w}_i^H \mathbf{h}_{bru} \right|^2 + \sigma_k^2 \leq \hat{\eta}_{p,k} \label{41c}}  
    \addConstraint{A^{*} / \hat{\kappa}_{c,k}^{(q)} - B^{*} \hat{\kappa}_{c,k} / \left( \hat{\kappa}_{c,k}^{(q)} \right)^2 \geq \hat{\nu}_{c,k}\label{41d}}       
    \addConstraint{A^{*} / \hat{\eta}_{p,k}^{(q)} - B^{*} \hat{\eta}_{p,k} / \left( \hat{\eta}_{p,k}^{(q)} \right)^2 \geq \hat{\beth}_{p,k}\label{41e}}  
    \addConstraint{\sum_{k=1}^K \hat{\wp}_{c,k} \leq \log_2 \left( 1 + \hat{\nu}_{c,k} \right), \quad \forall k \in K\label{41f}}
    \addConstraint{\log_2 \left( 1 + \hat{\beth}_{p,k} \right) \geq \hat{\Upsilon}_{p,k}, \quad \forall k \in K \label{41g}}
    \addConstraint{\left| \Theta_k \right| \leq 1, \quad k \in K \label{41h}}   
    \addConstraint{\eqref{40b}, \eqref{40c}, \eqref{40d}, \eqref{40e}}\notag
\end{maxi!}
\begin{algorithm}
\caption{AO for the Combined Solution}
\label{algo_2}
\begin{algorithmic}
\State \textbf{Input:} $\hat{\Theta}, \overline{\mathbf{W}}, S_{\text{sca}}^{\text{itr}}$ and \textbf{Initialize:} $\mathbf{W}^0, s=0$, and $R_{\text{sum}}^{-1}=0$
\While{$\left| R_{\text{sum}}^{(s)} - R_{\text{sum}}^{(s-1)} \right| \leq \epsilon$ or $s \leq S_{\text{sca}}^{\text{itr}}$}
\State Solve $\mathcal{P}4$ to obtain $\psi^{(s+1)}$
\State To obtain $\mathbf{W}^{(s+1)}$, solve $\mathcal{P}2$ using Algorithm 1 with obtained $\psi^{(s+1)}$ in step 2
\State Update $s = s + 1$
\EndWhile
\State \textbf{Output:} $R_{\text{sum}} = R_{\text{sum}}^{(s)}$
\end{algorithmic}
\end{algorithm}
where $A^{*}=2 \Re \left( \left( \bar{\mathbf{h}}_{bru}^{\mathbf{R}} \Theta^{(q)} \right) \bar{\mathbf{h}}_{bru}^{\mathbf{R}} \Theta \right)$, $B^{*}=\left| \bar{\mathbf{h}}_{bru}^{\mathbf{R}} \Theta^{(q)} \right|^2$,   $\hat{\mathbf{\kappa}}_p = [\hat{\kappa}_{p,1}, \cdots, \hat{\kappa}_{p,K}]$, $\hat{\mathbf{\eta}}_p = [\hat{\eta}_{p,1}, \cdots, \hat{\eta}_{p,K}]$, and $\bar{\mathbf{h}}_{bru}^{\mathbf{R}} = \mathbf{w}^H \mathbf{H}_A \mathbf{g}_{rk}$, $\forall k \in K$, and, $\Theta^{(q)}, \hat{\kappa}_{c,k}^{(q)}, \hat{\eta}_{p,k}^{(q)}$ are the values of $\Theta, \hat{\kappa}_{c,k}, \hat{\eta}_{p,k}$ obtained at the $q^{\text{th}}$ iteration, respectively. Note that the conditions in $\mathcal{P}4$ are convex and hence it is solved using CVX.
\subsection{Combined Approach and Computational Complexity}

Finally, we propose a method based on alternating optimization (AO) for designing of the joint optimization of the transmit beamforming and RIS beamforming, which is described in Algorithm~\ref{algo_2}. This approach iteratively addresses the subproblems of RIS beamformer design in $\mathcal{P}4$ and transmit beamforming in $\mathcal{P}2$. Until convergence within a tolerance $\epsilon$, the solution from each subproblem in one iteration is used as input for the next. While the precoder design in Algorithm~\ref{Algo_1} converges within $J_{\text{sca}}^{\text{itr}}$ iterations, we assume that the RIS phase-shift design converges within $Q_{\text{sca}}^{\text{itr}}$ SCA iterations. In $\mathcal{P}4$, the phase-shift design has $M_{BS}K+6K$ variables and $8N+1$ constraints. Therefore, the worst-case computational complexity for the corresponding subproblems is $\mathcal{O} \left( J_{\text{sca}}^{\text{itr}} (8K+1)^2 (M_{BS}K+6K) \right)$ and $\mathcal{O} \left( Q_{\text{sca}}^{\text{itr}} (8K+N)^2 (N+6K) \right)$. Finally, the overall computational complexity for Algorithm~\ref{algo_2} is calculated assuming that the AO algorithm converges in $S_{\text{sca}}^{\text{itr}}$ as
\begin{align}
&\mathcal{O} \left( S_{\text{sca}}^{\text{itr}} \left( J_{\text{sca}}^{\text{itr}} (8K+1)^2 (M_{BS}k+6k) + \right.\right.\nonumber\\
& \left.\left. Q_{\text{sca}}^{\text{itr}} (8K+N)^2 (N+6K) \right) \right).
\end{align}
\section{Numerical Results}
\begin{figure}[t!]
    \centering
        \includegraphics[width=\linewidth]{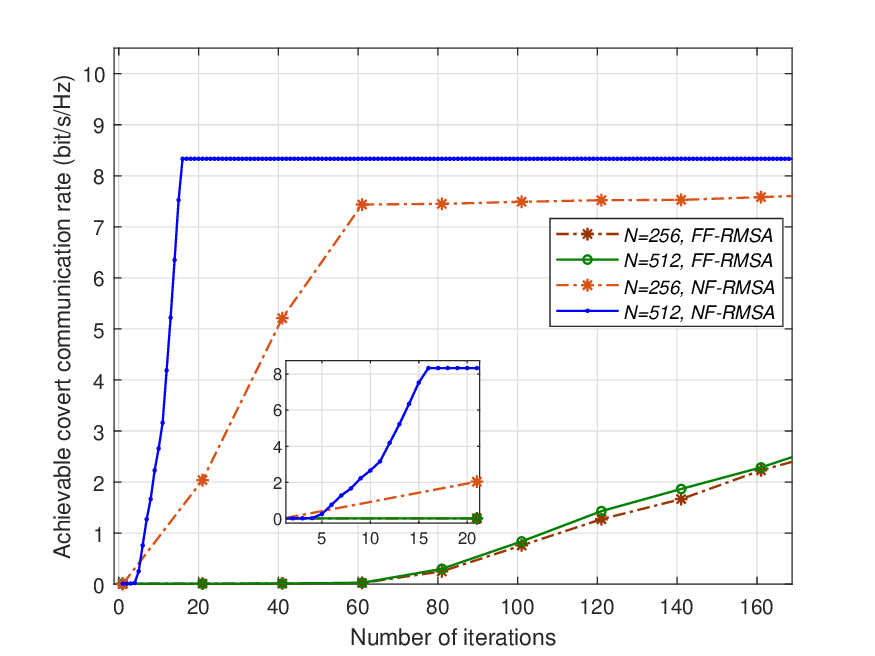}
    \caption{Convergence of the proposed algorithm.}
   \label{convergence}
\end{figure}
To evaluate the covert performance of the proposed NF RIS-aided RSMA system, simulations are conducted to assess its performance based on key performance metrics, namely the average covert communication rate, DEP, and OP.

These performance metrics are closely tied to the efficiency of the underlying optimization algorithm. Fig.~\ref{convergence} illustrates its convergence behavior under different RIS configurations, considering the system parameters $P_{bs,\max} = 30$ dBm, $M_{BS} = 4$, and $N = 512, 256$. As depicted in the figure, the algorithm converges within a few iterations. Specifically, it requires approximately 60 iterations for $N = 256$, while only 16 iterations are needed for $N = 512$, indicating a significantly faster convergence with a larger RIS. Moreover, a higher covert communication rate of approximately 8.2 bit/s/Hz is achieved for $N = 512$, compared to 7.3 bit/s/Hz for $N = 256$, corresponding to a 12.3\% improvement. This enhancement is attributed to the improved beamforming gain achieved in the NF region. \footnote{Such gains are not attainable in the FF region, which is included here solely for comparative analysis with the NF scenario.} Furthermore, despite optimized transmit beamforming and RIS phase shifts, the achievable covert communication rate remains notably low when both the legitimate users and Willie are located in the FF region, highlighting the performance advantage of the proposed NF-RSMA scheme. 
\begin{figure}[t!]
    \centering
        \includegraphics[width=\linewidth]{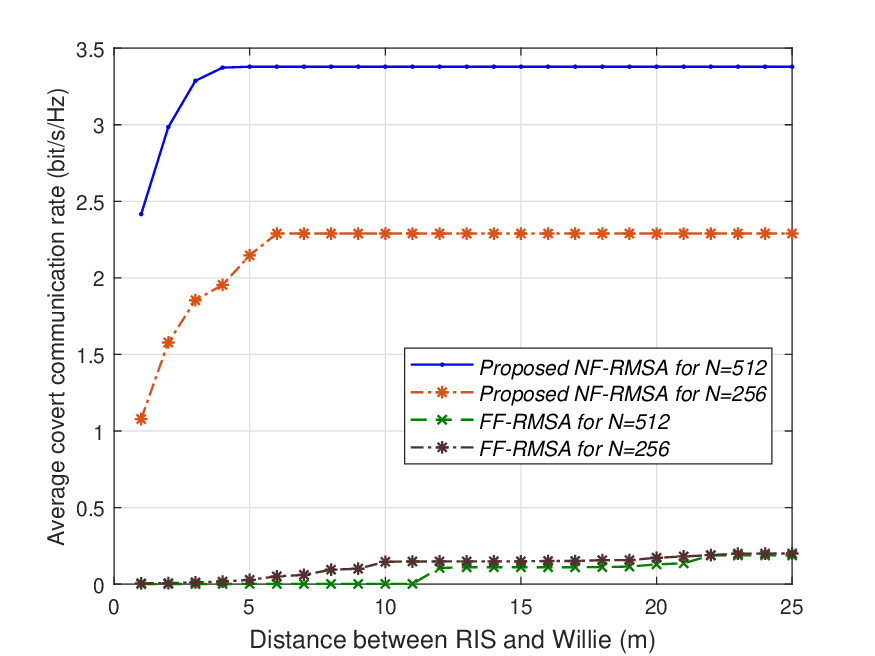}
    \caption{Average covert communication rate versus distance between RIS and Willie.}
   \label{distance}
\end{figure}
Fig.~\ref{distance} illustrates the effect of the separation distance between RIS and Willie on the covert communication rate for two different configurations of RIS elements. Notably, even when Willie is positioned close to Bob, the proposed algorithm maintains a considerable amount of covert communication rate across both RIS configurations. Additionally, as the distance increases, the covert rate improves significantly, highlighting the importance of the beam alignment in near-field communication enabled by optimized RIS phase shifts. In contrast, within the FF region, the covert communication rate approaches zero regardless of the number of RIS elements, indicating a negligible impact on Bob’s performance in such scenarios.
\begin{figure}[t!]
    \centering
        \includegraphics[width=\linewidth]{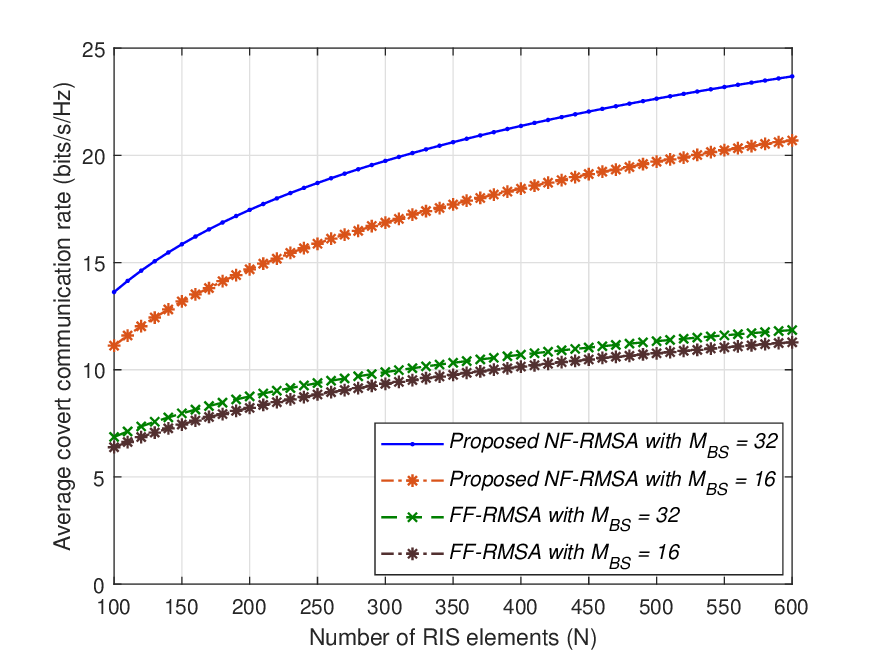}
    \caption{Average covert communication rate versus number of RIS elements ($N$) w.r.to different BS antennas ($M_{BS}$).}
   \label{RIS_elements}
\end{figure}
Fig.~\ref{RIS_elements} presents the average covert communication rate as a function of the number of RIS reflecting elements $N$, for two different base station antenna configurations, i.e., $M_{BS}=(16,32)$. As observed, the covert communication rate improves with an increase in both $N$ and  $M_{BS}$, indicating that larger RIS and broader BS apertures enhance system performance. Notably, the proposed algorithm demonstrates substantial gains in the NF scenario compared to the FF case. For instance, with $N = 350$  and $M_{BS} = 32$, the average covert communication rate reaches approximately 22~bits/s/Hz, which is 2.2 times higher than that in the FF scenario. This highlights the importance of RIS phase shift optimization in enhancing covert communication performance.


\begin{figure}[t!]
    \centering
        \includegraphics[width=\linewidth]{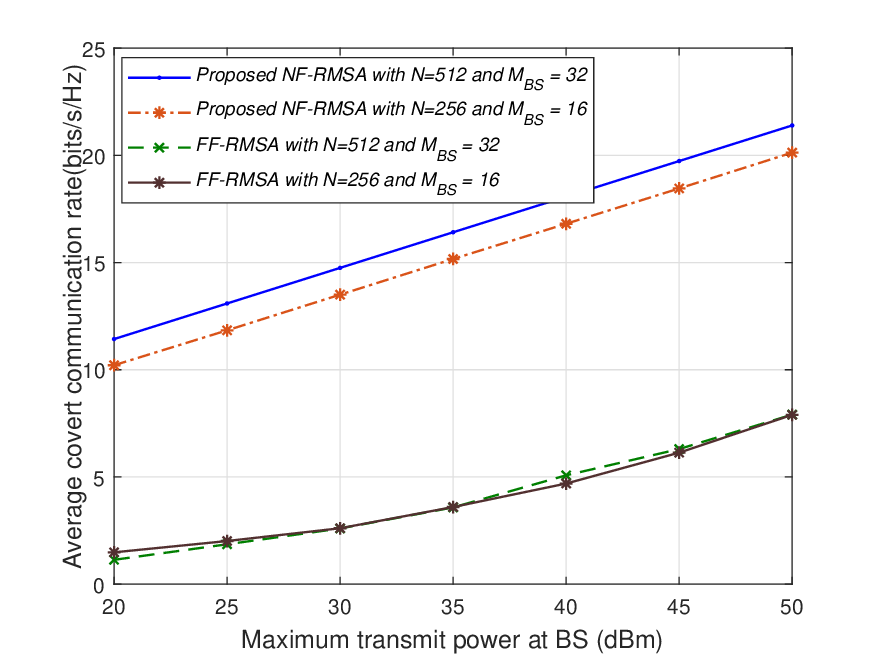}
    \caption{Average covert communication rate versus the maximum transmit power at the BS.}
   \label{BS_power}
\end{figure}
Fig.~\ref{BS_power} illustrates the average covert communication rate as a function of the maximum transmit power of the BS, denoted by $P_{bs,\max}$. As expected, increasing $P_{bs,\max}$ consistently leads to an enhanced average covert communication rate. For instance, at $P_{bs,\max} = 35$ dBm, the average covert communication rate of the proposed NF-RSMA reaches approximately 16~bits/s/Hz, which is four times higher than that achieved in the FF-RSMA scenario. This enhacement is due to the improved SINR at the users which is driven by the higher BS transmit power. In addition, the use of RSMA further contributes to this enhacement as it allows the BS to decompose the transmitted data into common and private components, which users can effectively recombine upon reception. Moreover, this splitting mechanism significantly boosts the achievable covert communication rate per user, especially at high transmit power levels, when combined with optimal reflection beamforming at the NF area. 

\begin{figure}[t!]
    \centering
        \includegraphics[width=\linewidth]{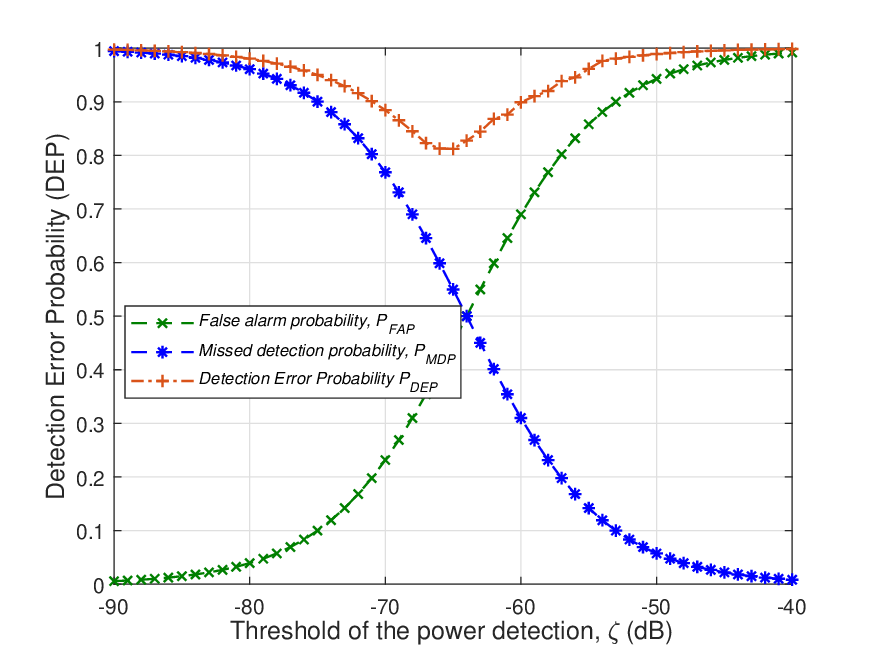}
    \caption{DEP versus Willie's detection power threshold, $\zeta$ (dB).}
   \label{DEP}
\end{figure}

Fig.~\ref{DEP} demonstrates the $P_{FAP}$, $P_{MDP}$, and $P_{DEP}$ versus Willie's detection power i.e., $\zeta$. As observed, $P_{FAP}$ decreases with increasing $\zeta$, while $P_{MDP}$ gradually increases until it reaches close to 1. Notably, the rate at which $P_{FAP}$ decreases with respect to $\zeta$ is more significant than the rate of increase in $P_{MDP}$, resulting in  a gradual decrease in the DEP. Moreover, as $\zeta$ increases further , $P_{FAP}$ tends toward zero, indicating that the system rarely encounters false alarms at higher thresholds. However, this also causes $P_{DEP}$ to rise, implying a trade-off where enhanced detection certainty compromises the level of covertness. Furthermore, for a fixed value of $\zeta$, Willie's $P_{FAP}$ increases with the BS's transmit power allocated to User~$K$. This implies that Willie becomes more likely to incorrectly infer the presence of BS's communication. 
\begin{figure}[t!]
    \centering
        \includegraphics[width=\linewidth]{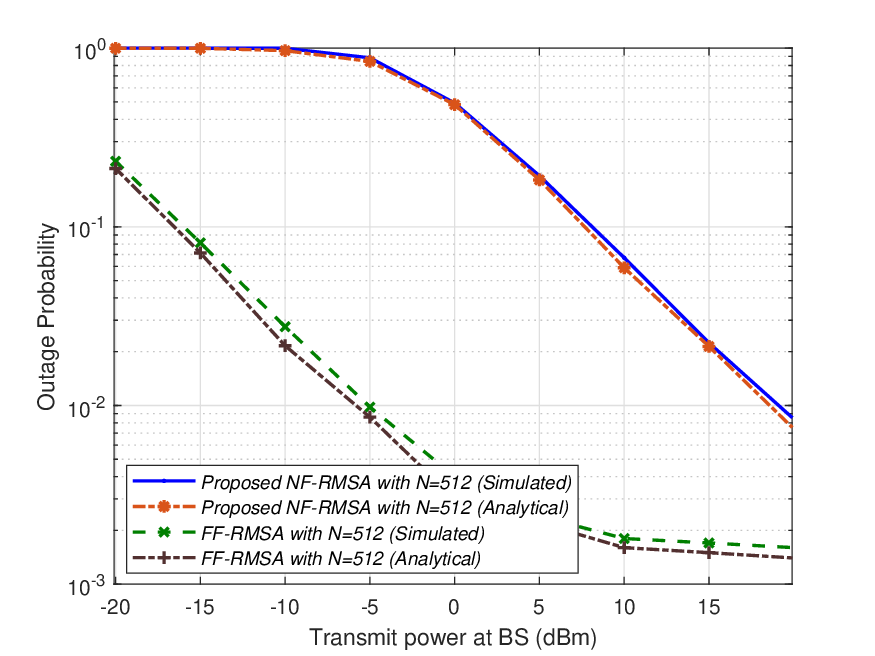}
    \caption{OP versus transmit power at the BS for $\gamma_{th}=-10dB$.}
   \label{op}
\end{figure}

Fig.~\ref{op} depicts the OP performance of Bob as function of the BS transmit power $P_{bs}$ with $\gamma_{th}=-10dB$. It is observed that the OP decreases with increasing $P_{bs}$, which can be attributed to the higher power allocated to Bob. Furthermore, it is observed that the asymptotic curve of the OP converges to both the theoretical and simulated results, demonstrating that our theoretical findings closely align with the simulations and effectively characterize the OP performance of the proposed NF-RSMA scheme. It is worth noting that the OP exhibits an inflection point as $P_{bs}$  reaches a particular threshold, this pattern can be explained by the reflected minimal beamforming during the initial OP decline. Although both scenarios utilize the same number of RIS elements, NF configuration yields significantly better performance compared to the FF case. This clearly demonstrates the superiority of the proposed NF-RSMA scheme in enhancing covert communication reliability.

\section{Conclusion and future work}
In this paper, we  proposed a novel framework for covert communication using RSMA-based NF RIS, where the average covert communication rate of Bob and that of other public users is optimized, while adhering to Willie's covertness limitation. An analytical framework was developed based on the SCA technique to address the joint optimization of the transmit beamforming at the base station and the reflection beamforming at the RIS through an AO algorithm. Simulation results demonstrate that the proposed RSMA-based NF scheme significantly outperforms the far-field scheme in terms of covert rate performance, even when Willie is positioned near both the RIS and Bob. Future work involves the analysis of this approach under imperfect CSI while considering scenarios with multiple Willies near legitimate users, which would critically affect overall performance.
\vspace{-1mm}
\section*{\uppercase{Appendix A}}

Since $\mathbf{g}_{rk}$ is deterministic in the NF LoS and $\mathbf{\Theta}$ is a diagonal phase shift matrix, therefore $\mathbf{\Theta} \mathbf{g}_{rk}$ is assumed to be scaled and phase-shifted version of $\mathbf{g}_{rk}$. Let us consider
\begin{align}\label{A.1}
    \mathbf{v}_k = \mathbf{\Theta} \mathbf{g}_{rk}   \tag{A.1},
\end{align}
where $v_{k,n} = \beta e^{j\theta_n} g_{rk,n}$. Then, \eqref{eqn_8} can be rewritten as
\begin{align}\label{A.2}
    \mathbf{h}_{bru,k} = \mathbf{H}_{br}^H \mathbf{v}_k.        \tag{A.2}
\end{align}
Since $\mathbf{H}_{br}$ is AWGN with $\mathbf{h}_{bru,k} \sim \mathcal{CN}(0, \sigma_{bru,k}^2 \mathbf{I}_{M_{BS}})$, where the variance is given as
\begin{align}\label{A.3}
    \sigma_{bru,k}^2 &= \sigma_{br}^2 \|\mathbf{v}_k\|^2= \sigma_{br}^2 |\beta|^2 \sum_{n=1}^N |g_{rk,n}|^2.     \tag{A.3}
\end{align}
Since $|g_{rk,n}|^2 = \frac{|\chi_{kn}|^2}{N} = \frac{1}{N (4\pi \frac{r_{kn}}{\lambda})^2}$, hence it can be expressed as
\begin{align}\label{A.4}
    \sum_{n=1}^N |g_{rk,n}|^2 = \sum_{n=1}^N \frac{1}{N (4\pi \frac{r_{kn}}{\lambda})^2} = \frac{1}{N} \sum_{n=1}^N \frac{1}{(4\pi \frac{r_{kn}}{\lambda})^2},        \tag{A.4}
\end{align}
and
\begin{align}\label{A.5}
    \sigma_{bru,k}^2 = \sigma_{br}^2 \beta^2 \frac{1}{N} \sum_{n=1}^N \frac{1}{(4\pi \frac{r_{kn}}{\lambda})^2}.         \tag{A.5}
\end{align}
Let us assume, $\lambda_k = \sum\limits_{n=1}^N \frac{1}{(4\pi \frac{r_{kn}}{\lambda})^2}$, therefore \eqref{A.5} can be modified as
\begin{align}\label{A.6}
    \sigma_{bru,k}^2 = \frac{1}{N}\beta^2\lambda_k\sigma_{br}^2.    \tag{A.6}
\end{align}

Consider the case where $\mathbf{w}_c$ and $\mathbf{w}k$ are designed to maximize the signal power at $\mathbf{h}_{bru,k}$. For simplicity, equal power allocation is assumed for the private messages, and a fraction of the total power is allocated to the common message, subject to the total power constraint given as
\begin{align}\label{A.7}
    \|\mathbf{w}_c\|^2 + \sum_{k=1}^K \|\mathbf{w}_k\|^2 \leq P_{bs}. \tag{A.7}
\end{align}
Additionally, consider $\|\mathbf{w}_c\|^2 = \alpha_c P_{bs}$, $\|\mathbf{w}_k\|^2 = \alpha_k P_{bs}$, with $\alpha_c + \sum\limits_{k=1}^K \alpha_k = 1$. Therefore, the approximate signal power for the common message can be expressed as
\begin{align}\label{A.8}
    |\mathbf{h}_{bru,k}^H \mathbf{w}_c|^2 \sim \sigma_{bru,k}^2 \|\mathbf{w}_c\|^2 \chi^{2}(2),   \tag{A.8}
\end{align}
where $\chi^2(2)$ is a chi-squared distribution having degrees of freedom of two that is obtained using the exponential distribution for the magnitude squared of the complex Gaussian random variable. The interference power can be given as
\begin{align}\label{A.9}
    \sum_{i=1}^K |\mathbf{h}_{bru,k}^H \mathbf{w}_i|^2 = \sum_{i=1}^K \sigma_{bru,k}^2 \|\mathbf{w}_i\|^2 \chi_i^2(2).   \tag{A.9}
\end{align}
Therefore, the SINR $\gamma_{c,k}$ can be given as
\begin{align}\label{A.10}
    \gamma_{c,k} = \frac{\sigma_{bru,k}^2 \alpha_c P_{bs} X_c}{\sum\limits_{i=1}^{K} \sigma_{bru,k}^2 \alpha_i P_{bs} X_i + \sigma_{k^2}},    \tag{A.10}
\end{align}
where $X_c, X_i \sim \text{Exp}(1)$.

Similarly, the SINR for the private message is given as
\begin{align}\label{A.11}
    \gamma_{p,k} = \frac{\sigma_{bru,k}^2 \alpha_k P_{bs} X_k}{\sum\limits_{i \neq k}^K \sigma_{bru,k}^2 \alpha_i P_{bs} X_i + \sigma_k^2}.    \tag{A.11}
\end{align}

Moreover, the CDF for $\gamma_{c,k}$ can be defined as
\begin{align}\label{A.12}
    F_{\gamma_{c,k}}(\gamma) = \Pr\left( \frac{\sigma_{bru,k}^2 \alpha_c P_{bs} X_c}{\sum\limits_{i=1}^{K} \sigma_{bru,k}^2 \alpha_i P_{bs} X_i + \sigma_{k^2}} \leq \gamma \right).   \tag{A.12}
\end{align}
Consider $Y = \sum\limits_{i=1}^K \alpha_i P_{bs} X_i + \frac{\sigma_k^2}{\sigma_{bru,k}^2}$, that implies $ F_{\gamma_{c,k}}(\gamma) = \Pr\left( X_c \leq \frac{\gamma Y}{\alpha_c P_{bs}} \right)$ which gives the equation in the form $ F_{X_c}(x) = 1 - e^{-x}$. Thus, it can be written as
\begin{align}\label{A.13}
    F_{\gamma_{c,k}}(\gamma) = \mathbb{E}_Y \left[ 1 - e^{-\frac{\gamma Y}{\alpha_c P_{bs}}} \right] = 1 - \mathbb{E}_Y \left[ e^{-\frac{\gamma Y}{\alpha_c P_{bs}}} \right].  \tag{A.13}
\end{align}
Therfore, the moment-generating function (MGF) of $Y$ can be expressed as
\begin{align}\label{A.14}
    \mathbb{E}\left[ e^{-t Y} \right] = e^{-t \frac{\sigma_k^2}{\sigma_{bru,k}^2}} \prod_{i=1}^K \frac{1}{1 + t \alpha_i P_{bs}},   \tag{A.14}
\end{align}
which gives the CDF as
\begin{align}\label{A.15}
    F_{\gamma_{c,k}}(\gamma) = 1 - e^{-\frac{\gamma \sigma_k^2}{\alpha_c P_{bs} \sigma_{bru,k}^2}} \prod_{i=1}^K \frac{1}{1 + \frac{\gamma \alpha_i}{\alpha_c}}.   \tag{A.15}
\end{align}

Similarly, for $\gamma_{p,k}$, the CDf can be expressed as
\begin{align}\label{A.16}
    F_{\gamma_{p,k}}(\gamma) = 1 - e^{-\frac{\gamma \sigma_k^2}{\alpha_k P_{bs} \sigma_{bru,k}^2}} \prod_{i \neq k}^K \frac{1}{1 + \frac{\gamma \alpha_i}{\alpha_k}}.  \tag{A.15}
\end{align}

\printbibliography 

\end{document}

%% file: acronyms.tex
\acrodef{NCR}{network-controlled repeater}
\acrodef{HO}[HO]{Handover}
\acrodef{UL}{uplink}
\acrodef{DL}{downlink}
\acrodef{3GPP}{3rd generation partnership project}
\acrodef{RIS}[RIS]{reconfigurable intelligent surface}
\acrodef{5G}{5th generation}
\acrodef{6G}[6G]{6th generation}
\acrodef{BS}[BS]{base station}
\acrodef{PL}[PL]{path-loss}
\acrodef{HOD}{HO decision} 
\acrodef{MIMO}{multiple input multiple output}
\acrodef{SNR}{signal-to-noise ratio}
\acrodef{SINR}[SINR]{signal-to-interference-plus-noise ratio}
\acrodef{SDP}{semidefinite programming}
\acrodef{CDF}{cumulative distribution function} 
\acrodef{CSI}{channel state information}
\acrodef{gNB}[BS]{base station}
\acrodef{AWGN}{additive white Gaussian noise}
\acrodef{CCI}{co-channel interference}
\acrodef{ULA}{uniform linear array}
\acrodef{RCS}{radar cross-section}
\acrodef{LoS}{line of sight} 
\acrodef{UE}[UE]{user equipment}
\acrodef{OFDM}{orthogonal frequency division multiplexing}
\acrodef{DoF}{degrees of freedom}
\acrodef{MI}{mutual interference}
\acrodef{Uu}{access link interface} 
\acrodef{TTT}{time-to-trigger}
\acrodef{HOM}{handover margin} 
\acrodef{RSRP}{received signal reference power}
\acrodef{RSRQ}[RSRQ]{received signal reference quality}
\acrodef{RLF}{radio link failure}
\acrodef{KPI}{key performance indicator}
\acrodef{HOC}{handover control}
\acrodef{HCP}{handover control parameter}
\acrodef{HOP}{handover probability} 
\acrodef{HOPP}{handover ping-pong} 
\acrodef{HCPs}{handover control parameters}
\acrodef{IoT}{internet of things}
\acrodef{V2X}{vehicle-to-everything}
\acrodef{UAV}{unmanned aerial vehicle}
\acrodef{AI}{artificial intelligence}
\acrodef{QoS}{quality of service}
\acrodef{THz}[THz]{terahertz}
\acrodef{NF}[NF]{near-field}
\acrodef{FF}[FF]{far-field}
\acrodef{UPW}[UPW]{uniform planewave}
\acrodef{NUSW}[NUSW]{non-uniform spherical wave}
\acrodef{SE}[SE]{ spectral efficiency}